\documentclass[reqno,centertags, 12pt]{amsart}
\usepackage{amsmath,amsthm,amscd,amssymb,mathtools}
\usepackage{esint}  
\usepackage{latexsym}
\usepackage{graphicx}
\usepackage{enumitem}
\usepackage[mathscr]{eucal}

\usepackage{hyperref}

\newcommand{\bbC}{{\mathbb{C}}}
\newcommand{\bbD}{{\mathbb{D}}}

\newcommand{\bbH}{{\mathbb{H}}}

\newcommand{\bbR}{{\mathbb{R}}}

\newcommand{\bbZ}{{\mathbb{Z}}}

\newcommand{\calA}{{\mathcal{A}}}

\newcommand{\calP}{{\mathcal P}}

\newcommand{\lb}{\label}

\newcommand{\supp}{\text{\rm{supp}}}

\newcommand{\bi}{\bibitem}

\newcommand{\beq}{\begin{equation}}
\newcommand{\eeq}{\end{equation}}
\newcommand{\ba}{\begin{align}}
\newcommand{\ea}{\end{align}}

\renewcommand{\Im}{\operatorname{Im}}
\renewcommand{\Re}{\operatorname{Re}}

\newcounter{smalllist}

\newcommand{\comm}[1]{}

\DeclareMathOperator{\diam}{diam}

\allowdisplaybreaks
\numberwithin{equation}{section}

\newtheorem{theorem}{Theorem}[section]
\newtheorem{proposition}[theorem]{Proposition}
\newtheorem{lemma}[theorem]{Lemma}

\theoremstyle{definition}
\newtheorem*{remark}{Remark}
\newtheorem*{remarks}{Remarks}

\newcommand{\jap}[1]{\langle #1 \rangle}
\newcommand{\bigjap}[1]{\left\langle #1 \right\rangle}
\newcommand{\Norm}[1]{\lVert#1\rVert}

\begin{document}

\title[Analyticity of Correlations]{The Strong Gauss Lucas Theorem and Analyticity of Correlation Functions via the Lee-Yang Theorem}
\author[B.~Simon]{Barry Simon$^{1}$}
\thanks{$^1$ Departments of Mathematics and Physics, Mathematics 253-37, California Institute of Technology, Pasadena, CA 91125, USA. E-mail: bsimon@caltech.edu.}

\dedicatory{Dedicated to the memory of Freeman Dyson.}

\

\date{\today}
\keywords{Ising Model, Correlation Functions, Lee-Yang, Gauss-Lucas, Cluster Expansions}
\subjclass[2020]{Primary: 82B20, 20C15, 30C15 ;Secondary: 41A58, 30B40}

\begin{abstract} We provide a simple mechanism for going from Lee-Yang type theorems to analyticity of correlation functions by exploiting under appreciated inequalities of Newman.  We also describe a Lee-Yang approach that recovers the consequences of a low density cluster expansion for spin $S$ models without any combinatorics.
\end{abstract}

\maketitle

\section{Introduction} \lb{s1}

Freeman Dyson was a master of a large swaths of modern theoretical and mathematical physics with important contributions.  He returned several times to the area I'd call the theory of lattice gases, i.e. the Ising and classical Heisenberg models.  Notable are his famous series \cite{DysLR, DysHM2, DysHM3} on the existence of phase transitions in slowly decaying $1D$ Ising models which also introduced the hierarchical models which turn out to be especially useful in mathematical understanding of the renormalization group.

On a more personal level, there is the joint papers \cite{DLS1, DLS2} he wrote with Elliott Lieb and me.  One of the high points of my time in Princeton \cite{IAMPPrinceton} were the weekly several hour meetings the three of us had in Freeman's office in the first few months of 1976 discussing many aspects of spin systems leading to our papers which contain what remains the only rigorous results on continuous system breaking in a quantum statistical mechanics model.  So it seemed appropriate to provide this memorial with some remarks on analyticity in classical lattice gases.

Our main subject here concerns proving analyticity of correlation functions of Ising models as a function of magnetic field using Lee-Yang methods.  This was first addressed by Lebowitz-Penrose \cite{LebPen1} in 1968 who were able to prove it in the spin $1/2$ case.  In 1974, Newman \cite{NewLY1} extended the Lee-Yang theorem to an optimal class of single spin distributions.  It appears that it wasn't until 2012 that Fr\"{o}hlich-Rodriguez \cite{FRod} proved analyticity of correlations in this generality; they had a second paper \cite{FRod2} on cluster expansions and decay of correlations in this generality.  One of our main points here is the remark that Newman \cite{NewLY1} could have proven this result by rather different methods using an inequality he proved but didn't use, namely for $\Re(h)>0$, one has that
\begin{equation}\label{1.1}
  \Re\left(\frac{f_\Lambda(j_1,\dots,j_n;h)}{f_\Lambda(j_1,\dots,j_{n-1};h)}\right) > 0
\end{equation}
where
\begin{equation}\label{1.2}
    f_\Lambda(j_1,\dots,j_n;h) \equiv \jap{\sigma_{j_1}\dots\sigma_{j_n}}_{\Lambda,h}
\end{equation}
with $\jap{\cdot}_{\Lambda,h}$ the free BC Ising expectation in magnetic field $h$.

The short version of this paper is the remark that while the Vitali convergence theorem \cite[Theorem 6.2.8]{BCA} is usually stated assuming the analytic functions, $g_n$, are uniformly bounded on compacts, it is valid if one merely has one-sided bounds on the real parts of $g_n$; the simplest way to see this is to note that if say $\Re(g_n)\ge 0$, then $h_n=e^{-g_n}$ are uniformly bounded, so we can apply Vitali to the $h_n$ and Hurwitz' Theorem \cite[Theorem 6.4.1]{BCA} to see that their limit is non-vanishing which implies convergence of the $g_n$.  From this observation, it is a few lines to conclude convergence and so analyticity of the correlation functions.

Rather than stop with this punchline and a really short paper, I plan to first provide, in Section \ref{s2}, the tools needed for a somewhat more direct proof of \eqref{1.1} and a version of the above observation with quantitative bounds (that thereby provides quantitative bounds on correlations).  In Section \ref{s3}, I will provide the details of the proof of Newman's result, \eqref{1.1} and of convergence and bounds on correlations.  Section \ref{s4} will address a related issue.  There exist (see that Section for references) an extensive literature on using Ruelle's extension of Asano's proof of the Lee-Yang theorem to obtain cluster expansions for Ising models in the high and low temperature regimes but there does not seem to be anything similar for the cluster expansion in the large field (aka large fugacity) region even though this doesn't require the somewhat involved group theoretic considerations of the work on high and low temperatures.  Since it is reasonable to have these results (which avoid any combinatorial estimates) in the literature, I sketch them in Section \ref{s4}.

\section{Fun and Games with Gauss Lucas} \lb{s2}

The Gauss-Lucas theorem (named after Lucas \cite{Lucas1, Lucas2, Lucas3}, whose earliest result was in 1868, and Gauss who never published it but had it in his letters and notebooks as early as 1835) asserts that if $P(z)$ is a polynomial, then the complex roots of $P'$, lie in the convex hull of the complex roots of $P$. The simplest proof follows from the formula
\begin{equation}\label{2.1}
  f(z) \equiv \frac{P'(z)}{P(z)} = \sum_{j=1}^{n}\frac{1}{z-z_j}
\end{equation}
if
\begin{equation}\label{2.2}
  P(z) = A\prod_{j=1}^{n}(z-z_j)
\end{equation}
from which the complex conjugate of $f(w)=0$ implies that (note that if $P'(w)=0$, then either $w$ is equal to some $z_j$ or else $f(w)=0$)
\begin{equation}\label{2.3}
  w=\sum_{j=1}^{n} a_jz_j; \qquad a_j\equiv \frac{|w-z_j|^{-2}}{\sum_{k=1}^{n} |w-z_k|^{-2}}
\end{equation}

The more common proof relies on what we'll call the strong Gauss-Lucas Theorem:

\begin{theorem} [Strong Gauss-Lucas Theorem] \lb{T2.1} If $P$ is a non-constant complex polynomial which is non-vanishing on $\bbH_+\equiv\{z\,\mid\,\Re(z)>0\}$, then one has that $f$ given by \eqref{2.1} obeys
\begin{equation}\label{2.4}
  z\in\bbH_+\, \Rightarrow \, \Re(f(z))>0
\end{equation}
In particular, $P'$ is non-vanishing on $\bbH_+$.
\end{theorem}

\begin{remark} Once one has this, one sees that if any given open half plane is free of zeros of $P$, it is free of zeros of $P'(z)$.  Since the convex hull of the zeros of $P$ is the complement of the union of all half planes free of zeros, this implies the Gauss-Lucas theorem.
\end{remark}

\begin{proof} If $\Re(z_j)\le 0$ and $z\in\bbH_+$, then $\Re(1/(z-z_j))>0$, so \eqref{2.1} implies \eqref{2.4}.
\end{proof}

Following Lieb-Sokal \cite{LiebSokal}, we define the space $\calA_a$ for any $a\ge 0$, as the space of entire functions with $\Norm{f}_b<\infty$ for all $b>a$ where
\begin{equation}\label{2.5}
  \Norm{f}_b = \sup_z e^{-b|z|^2}|f(z)|
\end{equation}
$\calA_a$ is a countable normed Fr\'{e}chet space with the set of norms $\Norm{f}_{a+1/n}$. One reason that it is better to deal with this Fr\'{e}chet space rather than the Banach space where is a single norm is finite is the freedom of being able to wiggle the value of $b$ in $\Norm{\cdot}_b$ gives us, as is seen by the following easy to prove fact

\begin{proposition} \lb{P2.2} (a) The Taylor series of any $f\in\calA_b$ (some $b\ge 0$) converge to $f$ in the topology of $\calA_b$.

(b) Let $\{f_m\}$ be a sequence which is bounded in $\calA_b$ (some $b\ge 0$), i.e., for each $c>b$, we have $\sup_m \Norm{f_m}_c < \infty$. Suppose the $f_m$ converges on a set with a limit point.  Then $f_m$ has a limit in $\calA_b$.
\end{proposition}

The spaces $\calA_a$ have analogs for functions of $\nu$ complex variables.  We define $\Norm{f}_b$ for entire functions, $f(z_1,\dots,z_\nu)$, of $\nu$ variables by
\begin{equation}\label{2.6}
  \Norm{f}_b = \sup\{e^{-b\sum_{j=1}^{\nu}|z_j|^2}|f(z)|\}
\end{equation}
$\calA_a(\bbC^\nu)$ is the space of functions with $\Norm{f}_b<\infty$ for all $b>a$. Proposition \ref{P2.2} extends easily to these spaces.

We also define $\calP^{\nu}$ to be the set of polynomials, $P(z_1,\dots,z_\nu)$, of $\nu$ variables which are non-vanishing if $Re(z_j)>0$ for $j=1,\dots,\nu$ (we denote this set of $\mathbf{z}$ by $\bbH_+^\nu$) and we let $\calP_a^{\nu}$ be its closure in $\calA_a(\bbC^\nu)\setminus\{f\equiv 0\}$.  If $\nu=1$, we will sometimes drop the superscript.  Since convergence in $\calA_a$ implies convergence uniformly on compacts, Hurwitz' theorem implies that if $f\in\calP_a^{\nu}$, then $f$ is non-vanishing on $\bbH_+^\nu$.  However we note that the converse is false for (see \cite{LiebSokal}) the function $z\mapsto e^{bz^2};b>0$ is non-vanishing on $\bbH_+$ but does not lie in any $\calP_a$.

One key to the proof of \eqref{1.1} will be (we use $\partial_j$ as shorthand for $\frac{\partial}{\partial z_j}$)

\begin{theorem} \lb{T2.3}  (a) For each $\nu$, $a>0$ and $j=1,\dots,\nu$, the map $\partial_j$ is a bounded map of $\calA_a(\bbC^\nu)$ to itself.

(b) If $f\in\calP^\nu_a$ and $\partial_j f$ is not identically zero, then $\partial_j f\in\calP^\nu_a$ for $j=1,\dots,\nu$.

(c) If $f\in\calP^\nu_a$ with $\partial_j f$ not identically zero, then on $\bbH_+^\nu$, we have that
\begin{equation}\label{2.7}
  \frac{\partial |f|^2}{\partial x_j}(z) = 2|f(z)|^2\Re\left(\frac{\partial_j f(z)}{f(z)}\right) >0
\end{equation}
\end{theorem}

\begin{proof} (a) By symmetry, we can suppose that $j=1$. By a Cauchy estimate,
\begin{equation}\label{2.8}
  \partial_1 f(\mathbf{z}) \le \Norm{f}_{b-\varepsilon,\nu}\exp\left((b-\varepsilon)\left[(|z_1|+1)^2+\sum_{k=2}^{\nu}|z_k|^2\right]\right)
\end{equation}
We can find $C$ so that for all $y>0$, we have that $2(b-\varepsilon)y\le \varepsilon y^2+C$, so with $G=\exp(C+(b-\varepsilon)^2)$, we have that
\begin{equation}\label{2.9}
  \partial_1 f(\mathbf{z}) \le G \Norm{f}_{b-\varepsilon,\nu}\exp\left(b\left[\sum_{k=1}^{\nu}|z_k|^2\right]\right)
\end{equation}

For any $b>a$, pick $\varepsilon=\tfrac{1}{2}(b-a)$ to get $\Norm{\partial_1 f(\mathbf{z})}_b \le G \Norm{f}_{b-\varepsilon,\nu}$ which proves (a).

(b) By the Gauss-Lucas theorem applied to the polynomial $P(\cdot,z_2,\dots,z_\nu)$, one sees that if $P\in\calP^\nu$, then so is $\partial_1 P$ so by (a), if $P_n\in\calP^\nu$ converges to $f$ in $\calA_a$, then $\partial_1 P_n$ converges to $\partial_1 f$ proving (b).

(c) By \eqref{2.4} applied to the polynomial $P(\cdot,z_2,\dots,z_\nu)$, one sees that if $P\in\calP^\nu$, then $\Re(\partial_1 P/P) >0$ on $\bbH_+^\nu$. Taking limits, one sees the final inequality in \eqref{2.7}.  For the first equality, we note that
\begin{align}\label{2.10}
  2\Re\left(\frac{\partial_j f(z)}{f(z)}\right) &= \frac{\partial_j f(z)}{f(z)}+ \frac{\overline{\partial_j f(z)}}{\overline{f(z)}} \nonumber \\
                                                &= \frac{\overline{f(z)}\partial_j f(z)+f(z)\overline{\partial_j f(z)}}{|f(z)|^2} \nonumber \\
                                                &= |f(z)|^{-2} (\partial_j+\overline{\partial_j})(f\overline{f}(z)) \\
                                                &= |f(z)|^{-2} \frac{\partial |f|^2}{\partial x_j}(z) \nonumber
\end{align}
where to get \eqref{2.10}, we used $\overline{\partial_j}f=\partial_j(\overline{f})=0$ (a form of the Cauchy-Riemann equations; see \cite[Problem 2.1.2]{BCA}).
\end{proof}

Later we will need a form of Theorem \ref{T2.3} for spin $1/2$ that goes back to the disk rather than the half plane.

\begin{lemma} \lb{L2.3A} Let $f(z)=Az+Bz^{-1}$ and suppose that for some $R>0$, we have that $f(z)\ne 0$ if $|z|<R$.  Then
\begin{equation}\label{2.10A}
  \Re\left(\frac{z\frac{\partial f}{\partial z}}{f(z)}\right)<0 \text{ and so } z\frac{\partial f}{\partial z}(z) \ne 0
\end{equation}
for $|z|<R$.
\end{lemma}

\begin{proof} The case $A=0$ is trivial.  If $A\ne 0$, the ratio in \eqref{2.10A} only depends on $B/A$, so without loss, we suppose that $A=1$.  In that case, the condition of not vanishing if $|z|<R$ is equivalent to $|B| \ge R^2$.  Noting that $z\tfrac{\partial f}{\partial z}=z-Bz^{-1}$, we compute
\begin{align}\label{2.10B}
   \Re\left(\frac{z\frac{\partial f}{\partial z}}{f(z)}\right) &=\Re \left[\frac{z-Bz^{-1}}{z+Bz^{-1}}\right] \nonumber \\
                   &=\frac{\Re[(z-Bz^{-1})(\bar{z}+\overline{B}\bar{z}^{-1})}{|z+Bz^{-1}|^2}
                   = \frac{|z|^2-|B|^2|z|^{-2}}{|z+Bz^{-1}|^2} <0
\end{align}
when $|z|<R$ since $|B|\ge R^2$.
\end{proof}

\begin{theorem} \lb{T2.3B} (a) Let $F(z_1,\dots,z_\nu)$ by a function on $(\bbC\setminus\{0\})^\nu$ of the form
\begin{equation}\label{2.10C}
  F(z_1,\dots,z_\nu) = \sum_{\sigma_1=\pm 1,\dots,\sigma_\nu=\pm 1} a(\sigma_1,\dots,\sigma_\nu) z_1^{\sigma_1}\dots z_\nu^{\sigma_\nu}
\end{equation}
and suppose that $F(z_1,\dots,z_\nu)\ne 0$ if $|z_1|<R_1,\dots,|z_\nu|<R_\nu$.  Then for any $\ell$ and $k_1,\dots,k_\ell\in\{1,\dots,\nu\}$, we have that $\prod_{j=1}^{\ell} z_{k_j}\frac{\partial }{\partial z_{k_j}}f(z) \ne 0$ if $|z_1|<R_1,\dots,|z_1|<R_\nu$ and on that set
\begin{equation}\label{2.10D}
  \Re\left(\frac{\prod_{j=1}^{\ell} z_{k_j}\frac{\partial }{\partial z_{k_j}}f(z)}
                  {\prod_{j=1}^{\ell-1} z_{k_j}\frac{\partial }{\partial z_{k_{j}}}f(z)}\right)<0
\end{equation}

(b) Let $Z$ be the function on $\bbC^\nu$ given by
\begin{equation}\label{2.10E}
  Z(h_1,\dots,h_n) = \sum_{\sigma_1=\pm 1,\dots,\sigma_\nu=\pm 1} a(\sigma_1,\dots,\sigma_\nu) \exp\left(\sum_{j=1}^{\nu} h_j\sigma_j\right)
\end{equation}
and suppose that $Z(h_1,\dots,h_\nu)\ne 0$ if $\Re(h_1)>A_1,\dots,\Re(h_\nu)>A_\nu$.  Then for any $\ell$ and $k_1,\dots,k_\ell\in\{1,\dots,\nu\}$, we have that $\frac{\partial^\ell}{\partial h_{k_1}\dots\partial h_{k_\ell}}Z(h) \ne 0$ if $\Re(h_1)>A_1,\dots,\Re(h_\nu)>A_\nu$ and on that set
\begin{equation}\label{2.10F}
  \Re\left(\frac{\frac{\partial^\ell}{\partial h_{k_1}\dots\partial h_{k_\ell}}Z(h)}
                  {\frac{\partial^{\ell-1}}{\partial h_{k_1}\dots\partial h_{k_{\ell-1}}}Z(h)}\right)>0
\end{equation}
\end{theorem}

\begin{remark} (b) is of course a consequence of the proof of Theorem \ref{T3.2} below in case the apriori measure is the spin $1/2$ Ising measure and one can get (a) from (b) by the change of variables we use to go in the other direction.  So this is an alternate proof of that special case.
\end{remark}

\begin{proof} (a) follows from the Lemma and induction since we can fix all the variables but the one we are taking the derivative of.  If $z_j=e^{-h_j}$, then $z_j\tfrac{\partial}{\partial z_j}=-\tfrac{\partial}{\partial h_j}$ so with $R_j=e^{-A_j}$, (a) implies (b).
\end{proof}

\bigskip

We end this section with the promised quantitative version of the remark on Vitali under only control of the real part (even though it has no relation to Gauss-Lucas).  One way is to use a basic result from complex analysis, the Borel-Carath\'{e}dory theorem \cite[Problem 3.6.12]{BCA}, that if $f$ is analytic on the unit disk $\bbD$ and continuous on its closure and $0<r<1$, then
\begin{equation}\label{2.11}
  \max_{|z|\le r}|f(z)| \le \frac{2r}{1-r}\max_{|z|= 1}\Re(f(z))+\frac{1+r}{1-r}|f(0)|
\end{equation}
From this and a simple covering argument, one easily shows if $z_0\in K\subset\Omega$ with $K$ compact and $\Omega$ open, there is a constant $C$ (depending only on $z_0$, $K$ and $\Omega$) so that for all $f$ analytic on $\Omega$, one has that
\begin{equation}\label{2.12}
  \sup_{z\in K} |f(z)| \le C\left(|f(z_0)|+\sup_{z\in\Omega} \Re(f(z))\right)
\end{equation}
Instead, we will use  the Herglotz representation for Carathe\'{e}dory functions \cite[Theorem 5.4.1]{HA}, i.e. an analytic function, $g$, on $\bbD$ with $g(0)=1$ and $\Re g\ge 0$, has the form:
\begin{equation}\label{2.13}
  g(z) = \int \frac{e^{i\theta}+z}{e^{i\theta}-z}\,d\mu(e^{i\theta})
\end{equation}
for a probability measure, $d\mu$, on $\partial\bbD$. Since $\max_{\theta}(|1+re^{-i\theta}|)=1+r$ and $\min_{\theta}(|1-re^{-i\theta}|)=1-r$, applying this to $g=f/f(0)$, one concludes that

\begin{theorem} \lb{T2.4} If $f$ is analytic on $\bbD$ with $\Re(f(z))>0$ there and $\Im(f(0))=0$, one has that
\begin{equation}\label{2.14}
  |f(z)| \le  f(0) \,\frac{1+|z|}{1-|z|}
\end{equation}
\end{theorem}

We will be interested in functions analytic with positive real part on $\bbH_+$, so we conformally map $\bbH_+$ to $\bbD$:

\begin{theorem} \lb{T2.5} If $f$ is analytic on $\bbH_+$ with $\Re(f(h))>0$ there and $\Im f(1)=0$, then for all $h\in\bbH_+$, we have that
\begin{equation}\label{2.15}
  \alpha(h)^{-1} f(1) \le |f(h)| \le\alpha(h) f(1); \qquad \alpha(h) \equiv \frac{|1+h|+|1-h|}{|1+h|-|1-h|}
\end{equation}
\end{theorem}

\begin{proof} Define
\begin{equation}\label{2.16}
  z(h) = \frac{1-h}{1+h}
\end{equation}
Then, $z$ maps the imaginary axis to the unit circle, has $z(1)=0$ and is a bijection of the Riemann sphere to itself so it maps $\bbH_+$ biholomorphically to $\bbD$.  Let $g$ be defined on $\bbD$ so that $g(z(h))=f(h)$.  Since $|z(h)|=|1-h|/|1+h|$, one sees that
\begin{equation}\label{3.6.10}
  \frac{1+|z|}{1-|z|} = \alpha(h)
\end{equation}
so the second inequality in \eqref{2.15} is just \eqref{2.14}.  By noting that $f(h)^{-1}$ also has a positive real part, we can apply the second inequality to $f(h)^{-1}$ to get the first inequality.
\end{proof}

\section{Convergence and Analyticity of Correlations} \lb{s3}

Here are the models we want to discuss.  We start with an even probability measure, $\mu$, on $\bbR$, called the apriori measure, which obeys

\begin{equation}\label{3.1}
  \int e^{Ax^2}\,d\mu(x) < \infty \text{ for all } A>0
\end{equation}
Given a finite set $\Lambda\subset\bbZ^\nu$, we let $\jap{\cdot}_{0,\Lambda}$ be the expectation in the product measure $\otimes_{k\in\Lambda}d\mu(x_j)$ on $\bbR$.  Fix a symmetric matrix $\{J_{k\ell}\}_{k,\ell\in\Lambda}$ with
\begin{equation}\label{3.2}
  J_{k\ell}\ge 0
\end{equation}
and form the Hamiltonian
\begin{equation}\label{3.3}
  H(\{x_k\}_{k\in\Lambda}) = -\sum_{k,\ell\in\Lambda} J_{k\ell}x_kx_\ell
\end{equation}

For $\mathbf{h}\in\bbC^{\Lambda}$, we are interested in the function (easily seen to be an entire function on $\bbC^\Lambda$):
\begin{equation}\label{3.4}
  Z(\mathbf{h}) = \jap{\exp(-H(x)-\sum_{k\in\Lambda}h_kx_k)}_{0,\Lambda}
\end{equation}
especially on the set $\bbH_+^\Lambda$.  Important is the Lee-Yang property of not vanishing on $\bbH_+^\Lambda$.  That this is true when $d\mu$ is the spin $1/2$ Ising measure, $\tfrac{1}{2}(\delta_{+1}+\delta_{-1})$, is the celebrated Lee-Yang circle theorem \cite{LY2} (the name comes from the fact that they used the variable $z=e^{2\beta h}$ for which their result shows what is in their case the polynomial $z^{|\Lambda|/2}Z(h_j\equiv h)$ has all its zeros on the unit circle).  The Lee-Yang theorem is important because of the realization of Lee-Yang \cite{LY1} that this property, convergence of $(Z_\Lambda)^{1/|\Lambda|}$ when all $h_j=h$ real and the Vitali theorem prove that the pressure (or free energy per unit volume depending on how the model is interpreted) is real analytic for $h>0$ and indeed has an analytic continuation to all of $\bbH_+$.

Newman \cite{NewLY1} found an optimal result specifying those $\mu$ for which the Lee-Yang property holds.  A PN measure is an even probability measure on $\bbR$ obeying the condition that
\begin{equation}\label{3.5}
  E_\mu(z) = \int\,e^{zx}\,d\mu(x)
\end{equation}
is non-vanishing whenever $\Re(z)>0$ (and so also when $\Re(z)<0$ since $E_\mu(-z)=E_\mu(z)$.).  I choose the name after Newman and P\'{o}lya \cite{Polya1, Polya2} (P\'{o}lya got interested in which even measures had Laplace transforms with only imaginary zeros as part of an unsuccessful attempt to prove the Riemann hypothesis).  In particular, P\'{o}lya proved that the measure $N^{-1} e^{-A\cosh(x)}\,dx$ is a PN measure.  P\'{o}lya's approach to the Riemann hypothesis was extended by deBruijn \cite{deBruijn} and Newman \cite{NewPolya}. Not all measures are PN measures; a direct calculation shows that the three point measure, $\tfrac{\lambda}{2}(\delta_{+1}+\delta_{-1})+(1-\lambda)\delta_0; 0<\lambda\le 1$ is a PM measure if and only if $\lambda\ge 1/3$.  But all measures of special interest in statistical mechanics are PN measures: this includes (equal weight) spin $S$ (either by a simple direct calculation or Griffiths \cite{GriffTrick}), the distribution of the first component of a unit vector equidistributed on a $D$-sphere (whose Fourier transform is well known to be a Bessel function, all of whose zeros are real), $N^{-1} e^{-A\cosh(x)}\,dx$ (done by P\'{o}lya, as noted), and $N^{-1}\exp(-ax^4+bx^2)\,dx$ (by Griffiths-Simon \cite{GriffSi} or as noted by Newman \cite{NewLY1} as a scaled limit of P\'{o}lya's example).

A moment's thought shows that a measure has the Lee-Yang property when all $J_{k\ell}=0$ if and only if it is a PN measure.  Newman \cite{NewLY1} made the remarkable discovery that this necessary condition for the Lee-Yang property for all ferromagnetic $J$ is also sufficient. Lieb-Sokal \cite{LiebSokal} found an alternate proof and more importantly the following stronger result.

\begin{theorem} \lb{T3.1} If $\mu$ is a PN measure, then for all $J_{k\ell}\ge 0$, the function $Z$ of \eqref{3.4} lies in $\calP^{|\Lambda|}_{a=0}$ and, in particular, is non-vanishing on $\bbH_+^{|\Lambda|}$.
\end{theorem}

\begin{remark} \cite{LiebSokal} prove a stronger result that obtains some results for apriori measures in some $\bbR^m$ but the quoted result has a simpler proof which they prove in an appendix (as well as from their more general result), and suffices for what we need here.
\end{remark}

We turn next to the correlation functions.  We fix a translation invariant pair interaction $J(j-k)\ge 0$ (with $J(-j)=J(j)$) and for $h>0$ define
\begin{equation}\label{3.6}
  f_\Lambda(j_1,\dots,j_n;h) \equiv \jap{\sigma_{j_1}\dots\sigma_{j_n}}_{\Lambda,h}
\end{equation}
where $\jap{\cdot}_{\Lambda,h}$ is the free BC state with pair interaction, external magnetic field $h$ and apriori measure $d\mu$ at each site.  We'll also define
\begin{equation}\label{3.7}
  J=\sum_{j\in\bbZ^\nu} J(j)
\end{equation}
which we suppose is finite.

It is a fundamental consequence of the analyticity guaranteed by Theorem \ref{T3.1} (Ruelle \cite{RuelleUnique}, Lebowitz-Martin-L\"{o}f \cite{LebML}) that when $d\mu$ is a PN measure, there is a unique equilibrium state which is the limit of the $\jap{\cdot}_{\Lambda,h}$. We use $f(j_1,\dots,j_n; h)$ for this limit. (The uniqueness result requires that $\supp(\mu)$ is compact.  In general \cite{LebPres}, one only gets a unique tempered state - the limit of the free BC is tempered.  Since it is peripheral, we'll ignore this issue; the reader can either supply details or assume the support is compact).

From Theorem \ref{T3.1} and the methods of Section \ref{s2}, we get the main result of this note:

\begin{theorem} \lb{T3.2} Let the single site distribution be a PN measure.  The infinite volume limits, $f(j_1,\dots,j_n;h)$, have analytic continuations to the region $\Re(h)>0$ and obey
\begin{equation}\label{3.8}
  \alpha(h)^{-n} L^n \le  |f(j_1,\dots,j_n;h)| \le \alpha(h)^n Q_n
\end{equation}
where $\alpha$ is given by \eqref{2.15}, $Q_n$ is an explicit $d\mu$ dependent constant,
\begin{equation}\label{3.9}
  L = \frac{\int x e^x\,d\mu(x)}{\int e^x,d\mu(x)}
\end{equation}
Moreover, the finite volume correlations converge to this analytic function for all $h$ with $\Re(h)>0$.
\end{theorem}

\begin{remarks} 1. We emphasize that the upper bounds in \eqref{3.8} depend only on $n$ and $h$ and are uniform the $j_k$'s.  This is useful in proving $m\ge 0$ for the argument in \cite{LebPen2}.  In general, we will get the upper bounds at $h=1$ using Ruelle \cite{RuUnbdd}, but, of course if $d\mu$ has compact support with convex hull, $[-S,S]$, we can use the use the trivial bound $Q_n=S^n$.

2. By the uniqueness of state for real $h$ we get that the limits exist for $\Re(h)>0$ for any BC where the Lee-Yang theorem is applicable (so the finite volume expectation has non-vanishing denominator), e.g. periodic BC.
\end{remarks}

\begin{proof} The derivatives of $Z_\Lambda(\{h_\ell\}_{\ell\in\Lambda})$ are given by
\begin{equation}\label{3.10}
  \frac{\partial^{n-1}}{\partial h_{j_1}\dots\partial h_{j_{n-1}}}Z = Z\jap{\sigma_{j_1}\dots\sigma_{j_{n-1}}}
\end{equation}
where the expectation is with a $j$ dependent $h_j$. So by Theorem \ref{T3.1} and \eqref{2.7}, we conclude inductively that $\jap{\sigma_{j_1}\dots\sigma_{j_{n-1}}}$ is non-vanishing when $h\in\bbH_+$ and in that region, one had that
\begin{equation}\label{3.11}
  \Re\left(\frac{f_\Lambda(j_1,\dots,j_{n};h)}{f_\Lambda(j_1,\dots,j_{n-1};h)}\right)>0
\end{equation}

By Theorem \ref{T2.5} and Vitali's theorem, the proof of our theorem is reduced to proving that
\begin{equation}\label{3.12}
   L^n \le  f_\Lambda(j_1,\dots,j_n;h=1) \le  Q_n
\end{equation}
By GKS inequalities (Kelly-Sherman \cite{KSonGKS} or Simon \cite[Chapter 2]{PTLG}), we have that $f_\Lambda(j_1,\dots,j_n;h=1) \ge \prod_{k=1}^{n} f_\Lambda(j_k;h=1) \ge (\jap{\sigma}_{0,h=1})^n=L^n$ where $\jap{\cdot}_{0,h=1}$ is the expectation of a single spin in external field $h=1$.  By Holder's inequality
\begin{equation*}
   f_\Lambda(j_1,\dots,j_n;h=1) \le \prod_{k=1}^{n}(f_\Lambda(j_k,\dots,j_k;h=1))^{1/n}
\end{equation*}
Under our assumptions, the system obeys all the requirement of Ruelle \cite{RuUnbdd} who proves \cite[Theorem 2.2]{RuUnbdd} explicit apriori bounds on probabilities that imply bounds on $f_\Lambda(j_k,\dots,j_k;h=1)$, uniformly in $\Lambda$; see \cite[Section 2.3]{PTLG}.
\end{proof}

\section{A Poor Person's Large Field Cluster Expansion} \lb{s4}

The Lee-Yang idea \cite{LY1} that tracking zeros can be used to prove analyticity can also be used to provide the results of cluster expansions without any combinatorial estimates at all.  As explicated by collaborations around Gruber and Slawny (some of their basic papers are \cite{GHM, MerGrub, SlawnyLT, SlawnyLT2, HSlawnyLT, SlawnyReview}) this can be done for both the high temperature and ferromagnetic low temperature regions.  Remarkably, there does not seem to be in the literature an explicit version of this for the large field (aka large fugacity, low density or high density) region even though as we'll see it is quite simple without the need for the involved group theoretic analysis of the high and low temperature exapnsions.  The one big limitation compared to the more usual cluster expansions is that the analysis is restricted to spin $1/2$ (or equal weight spin $S$).

Here is the framework we'll use.  At each point, $j\in\bbZ^\nu$, we have a $\pm 1$ Ising spin, $\sigma_j$.  For any finite subset $A\subset\bbZ^\nu$, we define
\begin{equation}\label{4.1}
  \sigma^A=\prod_{j\in A} \sigma_j
\end{equation}
and, as usual, $\jap{\cdot}_{0,\Lambda}$ is the product of equal weight Bernoulli expectation of spins in a finite set $\Lambda\subset\bbZ^\nu$.

We fix $J_0(A)\ge 0$ for all $A$ with $\#(A)\ge 2$ with two properties: it is translation invariant and the collection, $\calA$, of those $A$ with $J_0(A)\ne 0$ is finite range in the sense that
\begin{equation}\label{4.2}
  q=\#\{A\,\mid\,A\in\calA, A\ni 0\} < \infty
\end{equation}
We let $v$ be the number of equivalence class under translations of $A\in\calA$ (so $2v\le q$, since if $A\in\calA$, we have that $\#(A)$ translates of it containing $0$).

Below, when we write $J(A)$, we will mean possible complex numbers which are translation invariant (we will suppose that $J(A)=0$ if $A\notin\calA$). In the usual way (Ruelle \cite{RuBk1}, Simon \cite{SMLG}), for real parameters, $J(A)$ and $h$, one forms the finite volume Hamiltonian, partition function and pressure
\begin{align}\label{4.3}
  -H_\Lambda &= \sum_{A\subset\Lambda}J(A)\sigma^A+h\sum_{j\in\Lambda}\sigma_j \nonumber\\
        Z_\Lambda=\jap{e^{-H}}_{0,\Lambda}  &\qquad p=\lim |\Lambda|^{-1}\log(Z_\Lambda)
\end{align}
and one defines equilibrium states via the DLR equations. Finally, we define
\begin{equation}\label{4.4}
  I_0 = \max_{A\in\calA} 2^{\#(A)}e^{2J_0(A)}
\end{equation}

Here is what we'll prove:

\begin{theorem} \lb{T4.1} Given an interaction as just defined, there is a unique translation equilibrium invariant state and the pressure and all correlation functions are jointly analytic on the open set in $\bbC^{v+1}$ given by $\{h\,\mid\,|e^{-2h}|<1/qI_0\}\times\{J(A)\,\mid\, A\in\calA, |J(A)|<J_0(A)\}$.  For this unique translation invariant equilibrium state, the mass gap defined by
\begin{equation}\label{4.5}
  m = \limsup_{|k-\ell|\to\infty}\left\{-\frac{1}{|k-\ell|} \log
                        \left[\jap{\sigma_k\sigma_\ell} -\jap{\sigma_k}\jap{\sigma_\ell}\right]\right\}
\end{equation}
is strictly positive.
\end{theorem}

\begin{remarks} 1. More precisely, there is a unique translation invariant equilibrium state when all parameters are real and in the larger set, one has joint analyticity in the set described.

2. These are the major results one gets from cluster expansions in the large $h$ region for spin $1/2$. In particular, it has the initial results needed to prove joint analyticity in $(\beta,h)$ \cite{LebPen1} and mass gap \cite{LebPen2} in the full $\Re(h)>0$ for pair interacting ferromagnetic Ising models (see Simon \cite[Sections 3.7-3.8]{PTLG}).

3. By Griffiths \cite{GriffTrick}, a spin $S$ equal weight model has the same states,  pressure, etc as an analog spin $1/2$ model so by considering that model and fixing the coupling within the spin $1/2$ spins which sum to a spin $S$, one can extend this theorem to the spin $S$ situation.
\end{remarks}

We will use Ruelle's \cite{RuLY1, RuLY2} extension of the Asano contraction theorem \cite{Asano2}

\begin{proposition} [Ruelle-Asano Theorem] \lb{P4.2} Let $\Lambda$ be a finite set, let $z_{\Lambda}\equiv\{z_x\}_{x\in\Lambda}$ be the coordinates of a point in $\bbC^{|\Lambda|}$ and for $X\subset\Lambda$, define
\begin{equation}\label{4.6}
  z^X = \prod_{x\in X} z_x
\end{equation}
Let $\{\Lambda_\alpha\}_{\alpha\in A}$ be a finite cover of $\Lambda$ and for each $\alpha\in A$, a polynomial,
\begin{equation}\label{4.7}
   P_\alpha(z_{\Lambda_\alpha}) = \sum_{X\subset\Lambda_\alpha}c_X^{(\alpha)}z^X
\end{equation}
Suppose for each $\alpha\in A$ and $x\in\Lambda$, we have a closed set $M_x^{(\alpha)}\subset\bbC\setminus \{0\}$ so that if $z_x\notin  M_x^{(\alpha)}$ for all $x\in\Lambda_\alpha$, we have that $P_\alpha(z_{\Lambda_\alpha})\ne 0$.  Let
\begin{equation}\label{4.8}
  P(z_\Lambda) = \sum_{X\subset\Lambda}\left(\prod_{\alpha\,\mid\,\Lambda_\alpha\cap X\ne\emptyset}c_{\Lambda_\alpha\cap X}^{(\alpha)}\right)z^X
\end{equation}
Then $P(z_{\Lambda})\ne 0$ if for all $x\in\Lambda$, one has that
\begin{equation}\label{4.9}
  z_x\notin -\prod_{\alpha\,\mid\,x\in\Lambda_\alpha}(-M_x^{(\alpha)})
\end{equation}
\end{proposition}

Rather than make use of the formula for the final coefficients in \eqref{4.8}, we'll use the Asano contraction intuition that leads to it.  Multiaffine polynomials like \eqref{4.7} arise in the Lee-Yang \cite{LY1} scheme by replacing $h\sum_{j\in\Lambda}\sigma_j$ by $\sum_{j\in\Lambda}h_j\sigma_j$, multiplying $Z_\Lambda$ by $\exp(\sum_{j\in\Lambda}h_j\sigma_j)$ and writing the result as a function of $z_j=e^{2h_j}$.  Asano contraction results from forcing two spins, say, $\sigma_k$ and $\sigma_\ell$ to be parallel, i.e. dropping the terms with $\sigma_k=-\sigma_\ell$ and replacing $h_k\sigma_k+h_\ell\sigma_\ell$ by a single $h\sigma$ term.  Given the $P_\alpha$'s, one introduces variables $\{z_{\alpha,x}\}_{x\in\Lambda_\alpha}$, forms $\prod_\alpha P_\alpha(\{z_{\alpha,x}\}_{x\in\Lambda_\alpha})$ then contracts for each $x\in\Lambda$ pairwise all the $z_{\alpha,x}$ with $x\in\Lambda_\alpha$.  What results is \eqref{4.8} and the Proposition just tracks its zeros.

We'll need two lemmas to get the mass gap.

The first is an elementary piece of complex analysis:

\begin{lemma} \lb{L4.3} Let $f$ be analytic in $\bbD_R(0)$ with $C=\sup_{z\in\bbD_R(0)}|f(z)|<\infty$.  Suppose that $f^{(k)}(0)=0$ for $k=0,\dots,K-1$.  Then for any $z\in\bbD_R(0)$, we have that
\begin{equation}\label{4.10}
  |f(z)| \le \frac{C(|z|/R)^K}{1-(|z|/R)}
\end{equation}
\end{lemma}

\begin{proof} Let $f(z)=\sum_{n=0}^{\infty}a_n z^n$ be the Taylor expansion about $z=0$ for $f$. By a Cauchy estimate, \cite[Theorem 3.1.8]{BCA}
\begin{equation}\label{4.11}
  |a_n| \le CR^{-n}
\end{equation}
By the hypothesis on derivatives, the sum in the Taylor series starts at $n=K$.  We can sum the bounds in \eqref{4.11} using a geometric series to get \eqref{4.10}.
\end{proof}

The Ursell functions are defined by

\begin{equation}\label{4.12}
  u_n(X_1,X_2,\dots,X_n) = \left .\frac{\partial^n}{\partial h_1\cdots \partial h_n} \log\bigjap{\exp\left(\sum_{j=1}^{n}h_j X_j\right)}
                                                  \right|_{h_j=0}
\end{equation}

\begin{lemma} \lb{L4.4} Let $\jap{\cdot}$ be a product measure on single site distributions on $\Lambda\subset\bbZ^\nu$ (could be an infinite set).  For any finite set, $A\subset\Lambda$, let $\diam(A)=\max_{j\ne m\in A}|j-m|$.  Let $A_1,\dots,A_k$ be $k$ sets each with $\diam(A_j)\le R$.  Let $p,\ell\in\Lambda$ so that $|p-\ell| > kR$.  Then the Ursell function $u_{k+2}(\sigma_p,\sigma_\ell,\sigma^{A_1},\dots,\sigma^{A_k})= 0$.
\end{lemma}

\begin{proof} If $|p-\ell|>kR$, an easy geometric argument proves that one can break $\{1,\dots,k\}$ into two sets $P$ and $Q$ so that $\{p\}\cup_{j\in P} A_j$ is disjoint from $\{\ell\}\cup_{j\in Q} A_j$.  It is a basic fact (sometimes called the second Percus axiom (Percus \cite{Percus}) that if a set of variables can be broken into two independent pieces, its Ursell function vanishes; this follows immediately from \eqref{4.12} if we note that the $\log$ is a sum of $\log$s because of independence.  We conclude that the specified $u_{k+2}$ is zero.
\end{proof}

\begin{proof} [Proof of Theorem \ref{T4.1}] For each set finite $\Lambda\subset\bbZ^\nu$, we define (with $z^B=\prod_{j\in B}z_j$)
\begin{equation}\label{4.13}
  Z_{\Lambda,J(A)}(z)=\sum_{A\subset \Lambda } z^A p(A); \qquad   p(A)=\prod_{\mathclap{\substack{B\subset A \\ \#(B)>1,\text{ odd}}}}e^{-2J(B)}
\end{equation}
which with $z_j=e^{-2h_j}$ is the (analytic continuation of the) partition function in $j$ dependent field.  $A$ is the set of negative spins and $p(A)$ the Gibbs factor for $A$.

We model our proof on the Asano proof of the Lee-Yang circle theorem.  We define, for each $A$ with $A\in\calA$, as a polynomial of $|A|$ variables $\{z_{j,A}\}_{j\in A}$
\begin{equation}\label{4.14}
  R_A(\{z_{j,A}\}_{j\in A}) = \sum_{\mathclap{\substack{B\subset A \\ \text{ even}}}}\left(\prod_{j\in B}z_{j,A}\right)+
                                  e^{-2J(A)}\sum_{\mathclap{\substack{B\subset A \\ \text{ odd}}}}\left(\prod_{j\in B}z_{j,A}\right)
\end{equation}

The term for $B=\emptyset$ is $1$, and, if $r\equiv\max_j |z_{j,A}|\le 1$, all other terms are bounded by $e^{J_0(A)} r$, so $|R_A(z)-1| < 2^{|A|}e^{J_0(A)} r \le rI_0$, so if $r<1/I_0$, we see that $R_A(z)\ne 0$.

We can get $Z_\Lambda(\{z_j\}_{j\in\Lambda})$ by taking $\prod_{A\subset\Lambda} R_A(\{z_{j,A}\}_{j\in A})$ (i.e. copies of $z_j$ for each $A\ni j$) and Asano contracting all the copies of $z_{j,A}$ together to a single $z_j$.  Since there are at most $q$ copies of each $z_j$, we see that, by Proposition \ref{P4.2}, $Z$ is non-vanishing if all $z_j$ obey $|z_j|< 1/qI_0\equiv R_0$.  It will be important later (not to control pressure but to control correlation functions), that the same is true of the $Z_{per}$ defined with periodic BC, i.e. if we take $\Lambda$ to be a hypercube and take all $A$ with $J(A)\ne 0$ with $A\subset\Lambda$ and take all translates in $\Lambda$ when we connect it to a torus.

By the usual Lee-Yang argument \cite{LY1} (Vitali's convergence theorem), this proves joint analyticity of the pressure in the limit (for $z_j$ all taken equal).  Fix some real value of all $J(A)$ finite range. Now pick a $B$ with $B\notin\calA$ and form $Z_{B,J_B}(z)$.  When $J(B)=0$, $Z_B(z) = \prod_{j\in B}(1+z_j)$ is non-vanishing when all $z_j\in\bbD$.  By continuity of the zeros, it follows that for any $\rho<1$, we can find $\varepsilon(\rho)$ so that if $|J(B)|<\varepsilon(\rho)$, then $Z_{B,J_B}(z)\ne 0$ so long as $z\in\bbD_\rho(0)$.  If we now form $Z_\Lambda$ with an additional of $J(B)$ term, then $Z\ne 0$ so long as $|z_j|<\rho^{p}R_0=\rho^{p}/qI_0$ where $p=\#(B)$, for we get this from the $J(B)=0$ situation by an additional Asano contraction of all $Z_{C,J_B}(z)$ with $C$ a translation of $B$ that lies inside $\Lambda$ and each $z_j$ is involved in at most $p$ such contractions.  We thus see, because $\rho$ can be taken arbitrarily close to $1$, that there is a complex neighborhood (in $\bbC^{|\Lambda|+1}$) in the set of complex $z_j$ with $|z_j|<1/I_0q$ and $J(B)=0$ where the new $Z$ is non-vanishing.  It follows that for any real $h$ with $e^{-2h}<1/qI_0$, that the pressure is real analytic near $J(B)=0$, so, by \cite[Theorem III.3.11]{SMLG}, $\jap{\sigma^B}$ is the same in all translation invariant equilibrium states.  We also have analyticity in $J(A)$ where $J_0(A)\ne 0$.   Since $B$ was arbitrary, we conclude there is a unique such state!

Because any limit point of the periodic BC states is a translation invariant equilibrium state, we conclude that every such limit point is the unique translation invariant equilibrium state and thus we have convergence to that state.  By Theorem \ref{T2.3B}, we see that in the region of joint analyticity in $h$ and $J(A)$ described above, for any $j_1,\dots,j_\ell$ in a torus, $\Lambda$, we have that
\begin{equation}\label{4.15}
  \Re\left(\frac{\jap{\sigma_{j_1}\dots\sigma_{j_\ell}}_\Lambda}{\jap{\sigma_{j_1}\dots\sigma_{j_{\ell-1}}}_\Lambda}\right)>0
\end{equation}
where the expectation is with periodic BC.  Fix real values of the parameters. Since any limit point of the periodic BC state is a translation invariant equilibrium state and such a state is unique, we have convergence for such real parameters. By Theorem \ref{T2.5} and \eqref{4.15} by induction, we get uniform (in $\Lambda$) bounds on the correlation functions in the region of analyticity, so by Vitali's theorem we get convergence and analyticity of the infinite volume limit.

Fix some real $h$ in the region of analyticity.  Put a $\beta$ in front of all of the other couplings.  If we show $m>0$ for all small real positive $\beta$, then by the method of Lebowitz-Penrose \cite{LebPen2}, it is strictly positive everywhere in the region of analyticity. Because of Lemma \ref{L4.3} and the uniform bounds on $\jap{\sigma_k\sigma_\ell}-\jap{\sigma_k}\jap{\sigma_\ell}$ that follow from \eqref{4.15} and \eqref{2.15}, it suffices to prove that if, for some $R>0$, $|k-\ell|>QR$ implies that the $j$th derivative of $\jap{\sigma_k\sigma_\ell}-\jap{\sigma_k}\jap{\sigma_\ell}$ with respect to $\beta$ at $\beta=0$ vanishes for $j<Q$.  Since the limit of states at $\beta=0$ is uncoupled (and so independent) single sites, this follows with $R$ the range of the interaction because of Lemma \ref{L4.4} and \eqref{4.12} which shows that the $j$th derivative in question is a sum of Ursell functions of the form $u_{q+2}(\sigma_k,\sigma_\ell,\sigma^{A_1},\dots,\sigma^{A_k})$ with each $A_j\in\calA$.
\end{proof}

\begin{remark}  Instead of getting analyticity of correlations from \eqref{4.15}, one could use the proof of uniqueness of state - it implies joint analyticity of the pressure and so analyticity of the derivative which is the correlation.
\end{remark}


\end{document}